\documentclass[11pt]{amsart}
\usepackage{mathrsfs}
\usepackage{amssymb} %Extra symbols 
\usepackage{url}
\usepackage{hyperref}
\usepackage{graphicx} %Include postscript graphics

\theoremstyle{definition}
\newtheorem{theorem}{Theorem}[section]

\newtheorem{remark}[theorem]{Remark}

\def\({\left(}
\def\){\right)}

\newcommand{\R}{\mathbb{R}}
\newcommand{\N}{\mathbb{N}}

\newcommand{\de}{\textnormal{d}}

\newcommand{\ds}{\displaystyle}

\newcommand{\ie}{\textit{i.e.} }

\newcommand{\eg}{\textit{e.g.} }

\newcommand{\mc}[1]{\mathcal{#1}}

\newcommand{\rhord}[1]{\mc O_\rho\(#1\)}

\newcommand{\image}[3]{\begin{figure*}[ht]
\includegraphics[width=#2\textwidth]{#1}
\caption{\small{\label{#1}#3}}\end{figure*}}

\newcommand{\dsfrac}[2]{\ds{\frac{#1}{#2}}}

\newcommand{\schw}{Schwarzschild}
\newcommand{\rn}{Reissner-Nordstr\"om}
\newcommand{\kn}{Kerr-Newman}

\begin{document} 
 
%--------------------------------------------------------
% Title
\title[Kerr-Newman Solutions with Analytic Singularity]{Kerr-Newman Solutions with Analytic Singularity and no Closed Timelike Curves}

\author{Cristi \ Stoica}
\date{November 24, 2011}
\thanks{Partially supported by Romanian Government grant PN II Idei 1187.}

\begin{abstract}
It is shown that the Kerr-Newman solution, representing charged and rotating stationary black holes, admits analytic extension at the singularity. This extension is obtained by using new coordinates, in which the metric tensor becomes smooth on the singularity ring. On the singularity, the metric is degenerale - its determinant cancels. The analytic extension can be naturally chosen so that the region with negative r no longer exists, eliminating by this the closed timelike curves normally present in the Kerr and Kerr-Newman solutions. On the extension proposed here the electromagnetic potential is smooth, being thus able to provide non-singular models of charged spinning particles. The maximal analytic extension of this solution can be restrained to a globally hyperbolic region containing the exterior universe, having the same topology as the Minkowski spacetime. This admits a spacelike foliation in Cauchy hypersurfaces, on which the information contained in the initial data is preserved.
\bigskip
\noindent 
\end{abstract}

%--------------------------------------------------------
% Title and contents

\maketitle

\setcounter{tocdepth}{1}
\tableofcontents

\pagebreak
%~~~~~~~~~~~~~~~~~~~~~~~~~~~~~~~~~~~~~~~~~~~~~~~~~~~~~~~~~~~~~~~~~~~~~~~%
\section*{Introduction}

The {\kn} solutions are stationary and axisymmetric solutions of the Einstein-Maxwell equations, representing charged rotating black holes \cite{new65,Wal84}.

The other stationary black hole solutions can be obtained as particular cases of the {\kn} solutions. They are representative for all the black holes, because even the non-stationary black holes tend in time to {\kn} ones (according to the no-hair theorem).

But they also have some unusual properties, which are in general considered undesirable. They, as any black hole solution, have a singularity, where some of the fields reach infinite values. The singularity is in general ring-shaped, and passing through the ring one can reach inside another universe, in which there are \textit{closed timelike curves}, \ie time machines (which fortunately don't affect the causality in the region $r>0$) \footnote{The existence of closed timelike curves in the region $r<0$ of the {\kn} spacetime seems to depend on the coordinate system \cite{kim2002removal}.}. But there is also another problem, the \textit{black hole information paradox}, which refers to the loss of information inside the singularity, which, if would really happen, would cause serious problems, especially violation of unitary evolution, after the black hole evaporation \cite{Haw73,Haw76}.

The metric can be singular in two main ways which are relevant to our discussion. In the first kind of singularity, there are components of the metric which diverge as approaching the singularity. The {\kn} metric is, in usual coordinates, of the first kind. The second kind is that when the metric's components remain smooth at the singularity (and therefore finite). In the second kind, the singularity is still present\footnote{It may happen that the metric becomes regular after the coordinate transformation, but in this case it follows that the singularity was not genuine, it was due to the fact that the coordinates in which the regular metric was represented are singular. This is the case of the Eddington-Finkelstein coordinates, which proved that the singularity of the event horizon is only apparent.}, because the metric becomes \textit{degenerate} -- \ie its determinant becomes $0$. In some cases, it is possible to change the coordinate system in which a singularity of the first kind is represented, so that in the new coordinates the singularity becomes of the second kind --  it becomes degenerate.

The purpose of this article is to show that there are coordinates in which the  singularity of the {\kn} metric becomes of degenerate type. In these coordinates, the metric becomes smooth, and the only way the singularity manifests is that the metric becomes degenerate (we have already developed, in \cite{Sto11a,Sto11b,Sto11d}, mathematical tools which allow us to make differential geometry even in this situation of degenerate metric). In addition, we will show here that we can choose the analytic extension so that the closed timelike curves no longer exist. Moreover, we can find solutions which are globally hyperbolic and admit spacelike foliations in Cauchy hypersurfaces, ensuring therefore the conservation of information. The electromagnetic potential turns out to be smooth. New models for charged spinning particles are suggested. 

The {\kn} metric is usually defined in $\R\times\R^3$, where $\R$ is the time coordinate, and on $\R^3$ we use spherical coordinates $(r,\phi,\theta)$. Let $a\geq 0$ (which characterizes the rotation), $m\geq 0$ the mass, $q\in\R$ the charge, and let's define the functions
\begin{equation}
\label{eq_sigma}
\Sigma(r,\theta) := r^2 + a^2 \cos^2 \theta
\end{equation}
and
\begin{equation}
\label{eq_delta}
\Delta(r) := r^2 - 2 m r + a^2 + q^2.
\end{equation}
Then, we define the {\kn} metric by
\begin{equation}
\label{eq_g_t_t}
g_{tt} = - \frac{\Delta - a^2\sin^2\theta}{\Sigma}	
\end{equation}
\begin{equation}
\label{eq_g_r_r}
g_{rr} = \frac{\Sigma}{\Delta}	
\end{equation}
\begin{equation}
\label{eq_g_phi_phi}
g_{\phi\phi} = \frac{(r^2 + a^2)^2 -\Delta a^2\sin^2\theta}{\Sigma}\sin^2\theta
\end{equation}
\begin{equation}
\label{eq_g_t_phi}
g_{t\phi} = g_{\phi t} = - \frac{2a\sin^2\theta(r^2 + a^2 - \Delta)}{\Sigma}
\end{equation}
\begin{equation}
\label{eq_g_theta_theta}
g_{\theta\theta} = \Sigma
\end{equation}
all other components of the metric being equal to $0$ \cite{Wal84}.

By making $q=0$ we obtain the Kerr solution \cite{kerr63,ks65}, while by making $a=0$ we get the {\rn} solution \cite{reiss16,nord18}. By making both $q=0$ and $a=0$ we obtain the {\schw} solution, which when $m=0$ gives the empty Minkowski spacetime (see Table \ref{tab_static_black_holes}).
\begin{table}[htb!]
\centering
\begin{tabular}{|l|l|l|}
\hline
& $\mathbf{a>0}$ & $\mathbf{a=0}$ \\\hline
$\mathbf{q\neq 0}$ & {\kn} & {\rn} \\\hline
$\mathbf{q=0}$ & Kerr & {\schw} \\\hline
\end{tabular}
\caption{The various stationary black hole solutions, as particularizations of the {\kn} solution.}
\label{tab_static_black_holes}
\end{table}

%~~~~~~~~~~~~~~~~~~~~~~~~~~~~~~~~~~~~~~~~~~~~~~~~~~~~~~~~~~~~~~~~~~~~~~~%
\section{Extending the {\kn} spacetime at the singularity}
\label{s_kn_ext_ext}

\begin{theorem}
\label{thm_kn_ext_ext}
The {\kn} metric admits an analytic extension at $r=0$ (where the metric is degenerate, having analytic and not singular components).
\end{theorem}
\begin{proof}
We will find a coordinate system in which the metric is analytic, although degenerate. 
Recall that the event horizons of the black hole are given by the real solutions $r_\pm$ of the equation $\Delta=0$. It is enough to make the coordinate change in a neighborhood of the singularity -- in the block III, as it is usually called (\cite{ONe95}, p. 66). This is the region $r<r_-$ if $r_-$ is a real (and positive) number. If $\Delta=0$ has no real solutions, the singularity is naked, and we can take the entire domain.

We choose the coordinates $\tau$, $\rho$, and $\mu$, so that
\begin{equation}
\label{eq_coordinate_ext_ext}
\begin{array}{l}
\Bigg\{
\begin{array}{ll}
t &= \tau\rho^T \\
r &= \rho^S \\
\phi &= \mu\rho^M \\
\end{array}
\\
\end{array}
\end{equation}
with $S,T$ to be determined in order to make the metric analytic.
The expression of the metric tensor when passing from coordinates $(x^{a})$ to the new coordinates $(x^{a'})$ is given by:
\begin{equation}
\label{eq_metric_coord_change}
g_{a'b'} = \dsfrac{\partial x^{a}}{\partial x^{a'}}\dsfrac{\partial x^{b}}{\partial x^{b'}}	g_{ab}
\end{equation}
where Einstein's summation convention is used.
In our case, we have the following Jacobian for the coordinate transformation:
\begin{equation}
\label{eq_coordinate_change}
\begin{array}{l}
\dsfrac{\partial(t,r,\phi,\theta)}{\partial(\tau,\rho,\mu,\theta)}=
\left(
\begin{array}{llll}
\dsfrac{\partial t}{\partial \tau} & \dsfrac{\partial t}{\partial \rho} & \dsfrac{\partial t}{\partial \mu} & \dsfrac{\partial t}{\partial \theta}  \\
\dsfrac{\partial r}{\partial \tau} & \dsfrac{\partial r}{\partial \rho} & \dsfrac{\partial r}{\partial \mu} & \dsfrac{\partial r}{\partial \theta}  \\
\dsfrac{\partial \phi}{\partial \tau} & \dsfrac{\partial \phi}{\partial \rho} & \dsfrac{\partial \phi}{\partial \mu} & \dsfrac{\partial \phi}{\partial \theta}  \\
\dsfrac{\partial \theta}{\partial \tau} & \dsfrac{\partial \theta}{\partial \rho} & \dsfrac{\partial \theta}{\partial \mu} & \dsfrac{\partial \theta}{\partial \theta}  \\
\end{array}
\right)
=
\left(
\begin{array}{llll}
\rho^T & T \tau \rho^{T-1} & 0 & 0 \\
0 & S\rho^{S-1} & 0 & 0  \\
0 & M \mu \rho^{M-1} & \rho^M & 0 \\
0 & 0 & 0 & 1 \\
\end{array}
\right)
\end{array}
\end{equation}
Let's arrange its coefficients in a table:
\begin{table}[htb!]
\centering
\begin{tabular}{|p{1.5 cm}||p{1.5 cm}|p{1.5 cm}|p{1.5 cm}|p{1.5 cm}|}
\hline
& $\cdot /\partial \tau$ & $\cdot /\partial \rho$ & $\cdot /\partial \mu$ & $\cdot /\partial \theta$ \\\hline
\hline
$\partial t/\cdot$ & $\rho^T$ & $T \tau \rho^{T-1}$ & $0$ & $0$ \\\hline
$\partial r/\cdot$ & $0$ & $S\rho^{S-1}$ & $0$ & $0$ \\\hline
$\partial \phi/\cdot$ & $0$ & $M \mu \rho^{M-1}$ & $\rho^M$ & $0$ \\\hline
$\partial \theta/\cdot$ & $0$ & $0$ & $0$ & $1$ \\\hline
\end{tabular}
\caption{The Jacobian coefficients of the coordinate change.}
\label{tab_jacobian}
\end{table}

We want to make sure that the new expression of the metric becomes smooth even on the ring singularity. For this, we want that all the terms in the right hand side of equation \eqref{eq_metric_coord_change} are smooth. To ensure this, we have to make sure that the Jacobian coefficients cancels the singularities of the metric components, even when $\cos \theta=0$.

The least power of $\rho$ on the ring singularity, in each of the metric components listed in equations \eqref{eq_g_t_t}, \eqref{eq_g_phi_phi}, and \eqref{eq_g_t_phi} are respectively:
\begin{equation}
\label{eq_ord_g_t_t}
\rhord{g_{tt}} = - 2S
\end{equation}
\begin{equation}
\label{eq_ord_g_phi_phi}
\rhord{g_{\phi\phi}} = -2S
\end{equation}
\begin{equation}
\label{eq_ord_g_t_phi}
\rhord{g_{t\phi}} = \rhord{g_{\phi t}} = - 2S
\end{equation}
these components being obtained by dividing polynomial expressions in $\rho$ by $\Sigma$. None of the other components can become singular on the ring singularity.

The least power of $\rho$ in each of the Jacobians coefficients from Table \ref{tab_jacobian} are given in Table \ref{tab_jacobian_ord}.
\begin{table}[htb!]
\centering
\begin{tabular}{|p{1.5 cm}||p{1.5 cm}|p{1.5 cm}|p{1.5 cm}|p{1.5 cm}|}
\hline
& $\cdot /\partial \rho$ & $\cdot /\partial \tau$ & $\cdot /\partial \mu$ & $\cdot /\partial \theta$ \\\hline
\hline
$\partial t/\cdot$ & $T$ & $T-1$ & $0$ & $0$ \\\hline
$\partial r/\cdot$ & $0$ & $S-1$ & $0$ & $0$ \\\hline
$\partial \phi/\cdot$ & $0$ & $M-1$ & $M$ & $0$ \\\hline
$\partial \theta/\cdot$ & $0$ & $0$ & $0$ & $0$ \\\hline
\end{tabular}
\caption{The least power of $\rho$ in the Jacobian coefficients of the coordinate change.}
\label{tab_jacobian_ord}
\end{table}

Let's take the metric components and see if they are canceled by the coefficients of the Jacobian.

We check each component $g_{ab}$ of the metric tensor by looking up the rows labeled by $\partial x^a/\cdot$ and $\partial x^b/\cdot$ in Table \ref{tab_jacobian_ord}.

For example, the term
\begin{equation}
	\dsfrac{\partial t}{\partial \rho}\dsfrac{\partial t}{\partial \tau}g_{tt}
\end{equation}
satisfies
\begin{equation}
	\rhord{\dsfrac{\partial t}{\partial \rho}\dsfrac{\partial t}{\partial \tau}g_{tt}} = (T  - 1) + T - 2S
\end{equation}
hence $T$ needs to satisfy $2T\geq 2S+1$.

From equations \eqref{eq_ord_g_t_t}, \eqref{eq_ord_g_phi_phi}, and \eqref{eq_ord_g_t_phi} is easy to see that we have to do this only for the components of the metric with indices $t$ and $\phi$. From the equation \eqref{eq_metric_coord_change} we see that we are interested only in the rows $\partial t/\cdot$ and $\partial \phi/\cdot$ from the Table \ref{tab_jacobian_ord}. It follows then that each of the coefficients of the Jacobian having the form $\partial t/\cdot$ and $\partial \phi/\cdot$ has to contain $\rho$ to at least the power $S$, to cancel the metric components. It follows that the conditions
\begin{equation}
\label{eq_metric_smooth_cond}
\begin{array}{l}
\Bigg\{
\begin{array}{ll}
S &\geq 1 \\
T &\geq S + 1 \\
M &\geq S + 1	
\end{array}
\\
\end{array}
\end{equation}
where $S,T,M\in\N$, ensure the smoothness (and the analyticity for that matter) of the metric on the ring singularity, in the new coordinates.
None of the metric components in the new coordinates become infinite at the singularity.
\end{proof}

\image{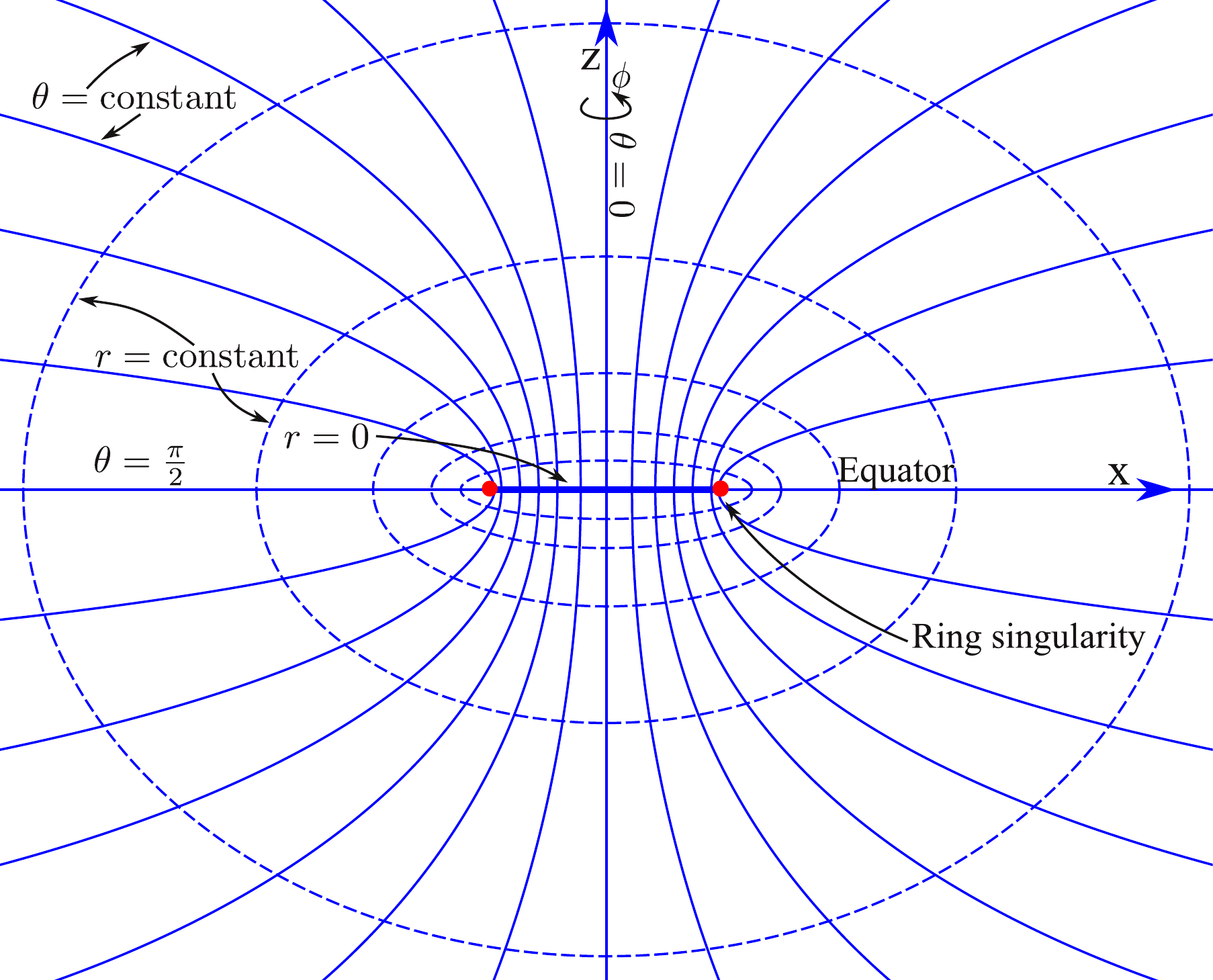}{0.95}{The Kerr (and \kn) solution, in Kerr-Schild coordinates. The standard solution admits an analytic continuation beyond the disk $r=0$, into another spacetime which contains closed timelike curves. If we take in our solution $S$ to be even, we can identify isometrically the regions $\rho<0$ and $\rho>0$, and obtain by this a removal of the wormhole and of the closed timelike curves.}

\begin{remark}
\label{rem_kn_no_wormhole}
The {\kn} solution has a ring singularity, where $r=0$ and $\cos\theta=0$. By using Kerr-Schild coordinates, we can see that it can be analytically extended through the disk defined by $r=0$ to another spacetime region which looks similar, but is not isometric to the region with $r>0$, since there $r<0$ (see Fig. \ref{kerr-schild}). On the other hand, if we use our coordinates with even $S$, then the metric becomes even in $\rho$, and the analytic extension to $\rho<0$ gives a region which is isometric to that with $\rho>0$. We can isometrically identify these regions, by identifying the points $(\rho,\tau,\mu,\theta)$ and $(-\rho,\tau,\mu,\theta)$.
\end{remark}

\begin{remark}
\label{rem_kn_no_ctc}
Our global solution described in the Remark \ref{rem_kn_no_wormhole} shows that, for even $S$, we no longer have regions where $r<0$. In this case, the closed timelike curves known to appear in the standard Kerr and {\kn} solutions, are no longer present. Therefore, if these closed timelike curves were considered as violating the causality, to avoid them we just take $S$ to be even and make the identification of $(\rho,\tau,\mu,\theta)$ and $(-\rho,\tau,\mu,\theta)$.
\end{remark}

\begin{remark}
\label{rem_kn_vs_rn}
If $a\to 0$, then we recover the {\rn} solution. The neck $r=0$ connecting the two regions $r>0$ and $r<0$ converges to a point, as well as the ring singularity delimiting it. This point is the $r=0$ singularity of the {\rn} solution, and it still can be viewed as connecting the region $r>0$ with a region $r<0$. This can be now put in relation with the extension through singularity of some of the {\rn} solution developed in \cite{Sto11f}, which suggest that for odd $S$ the singularity connects the spacetime region $r>0$ with a region $r<0$.
\end{remark}

%~~~~~~~~~~~~~~~~~~~~~~~~~~~~~~~~~~~~~~~~~~~~~~~~~~~~~~~~~~~~~~~~~~~~~~~%
\section{The electromagnetic field}
\label{s_kn_em}

One distinctive feature of our extension is that it has smooth electromagnetic potential and electromagnetic field. This may be important in particular when using the {\kn} black holes to model charged particles.

The electromagnetic potential of the {\kn} solution is
\begin{equation}
\label{eq_kn_electromagnetic_potential}
A = -\dsfrac{qr}{\Sigma}(\de t - a\sin^2\theta\de\phi)
\end{equation}
which becomes in our coordinates
\begin{equation}
\label{eq_kn_electromagnetic_potential_smooth}
A = -\dsfrac{q\rho^S}{\Sigma}(\rho^T\de\tau + T\tau\rho^{T-1}\de\rho - a\sin^2\theta\rho^M\de\mu)	
\end{equation}
because from the Table \ref{tab_jacobian} it follows that
\begin{equation}
\label{eq_de_t}
	\de t = \rho^T\de\tau + T\tau\rho^{T-1}\de\rho
\end{equation}
\begin{equation}
\label{eq_de_r}
	\de r = S\rho^{S-1}\de\rho
\end{equation}
and
\begin{equation}
\label{eq_de_phi}
	\de \phi = M\mu\rho^{M-1}\de\rho + \rho^M\de\mu
\end{equation}

The singularity of the electromagnetic potential $A$ at $\rho=0$ and $\cos\theta=0$ is removed in our case, since  $T> S$ and $M> S$, from the conditions \eqref{eq_metric_smooth_cond}. Similarly, since $F = \de A$, we conclude that the electromagnetic field is smooth too.

%~~~~~~~~~~~~~~~~~~~~~~~~~~~~~~~~~~~~~~~~~~~~~~~~~~~~~~~~~~~~~~~~~~~~~~~%
\section{The global solution}
\label{s_kn_global}

The Penrose-Carter diagrams of our solution depend on the various combinations of the parameters $a,q,m$. For the {\schw} solution they were presented in \cite{Sto11e}, and for the {\rn} in \cite{Sto11f}. In general it is admitted that the Kerr and {\kn} solutions have Penrose-Carter diagrams similar to those for the {\rn} solution, although there are some differences due to the fact that the symmetry is not spherical, but axisymmetric, that the singularity is ring-shaped, and of the closed timelike curves in the region $r<0$. Since our solution can eliminate the closed timelike curves (Remark \ref{rem_kn_no_wormhole}), we expect a better similarity with the {\rn} case, and consequently similar Penrose-Carter diagrams. This would allow similar spacelike foliations of the spacetime as those presented in \cite{Sto11f} for the {\rn} case, except that the singularity is ring-shaped (see Figure \ref{kn-ext}). The foliations are obtained exactly as in the {\rn} case \cite{Sto11f}, by using the same Schwarz-Christoffel mappings. As in that case, to obtain maximal globally hyperbolic extensions, we don't take the maximal analytic continuations of the solutions for $a^2 + q^2 \geq m^2$ beyond the Cauchy horizons. To avoid these horizons, we limit the foliations to globally hyperbolic regions containing the exterior universe.

\image{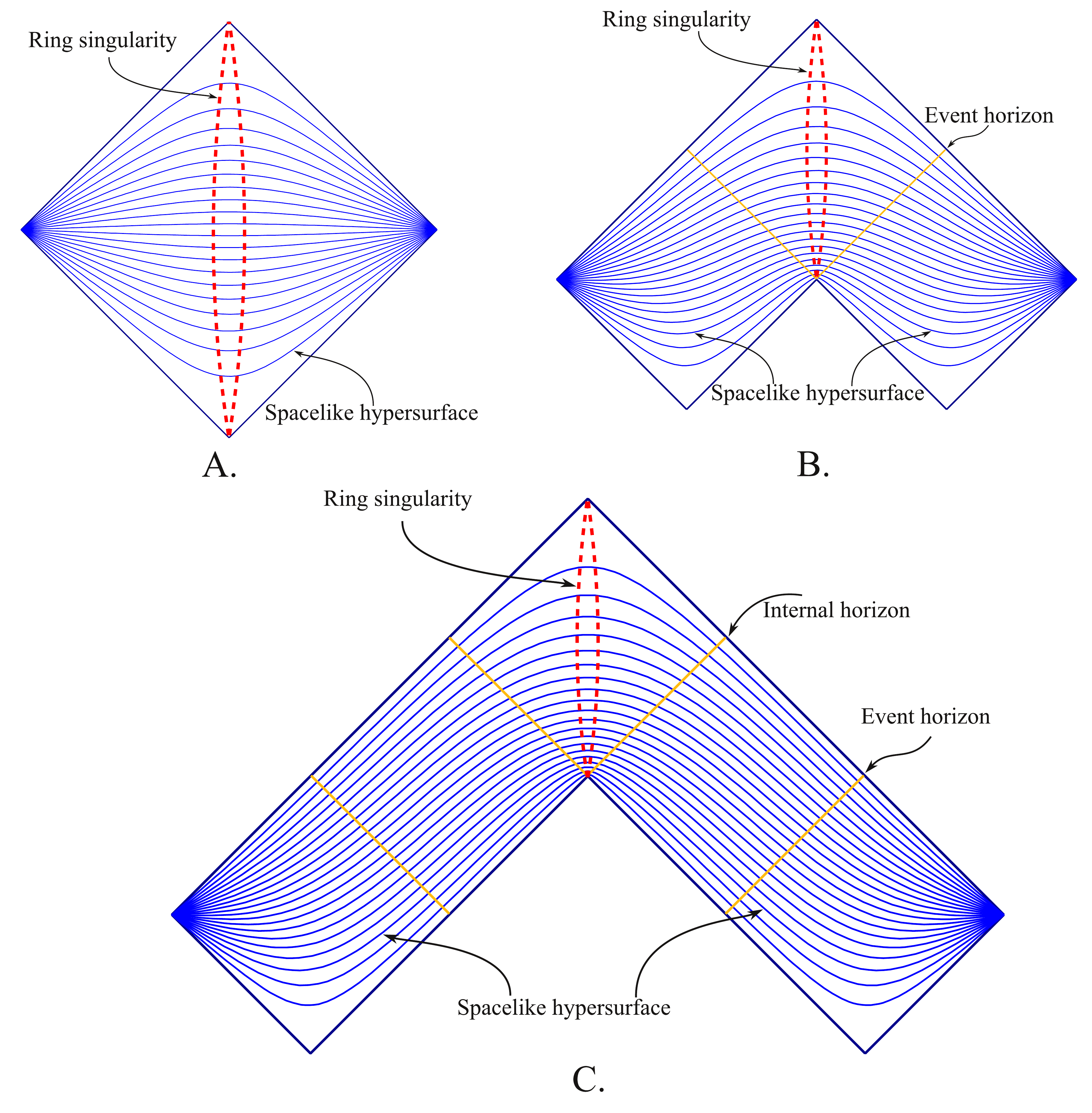}{0.95}{\textbf{A.} Space-like foliation of the naked {\kn} solution ($a^2+q^2>m^2$). \textbf{B.} Space-like foliation of the extremal {\kn} solution with $a^2+q^2=m^2$. \textbf{C.} Space-like foliation of the non-extremal {\kn} solution ($a^2+q^2<m^2$).}

%~~~~~~~~~~~~~~~~~~~~~~~~~~~~~~~~~~~~~~~~~~~~~~~~~~~~~~~~~~~~~~~~~~~~~~~%
\section{The significance of the analytic extension at the singularity}
\label{s_kn_analytic_significance}

The analytic extension beyond the singularity obtained here completes the series of results obtained for the {\schw} \cite{Sto11e} and {\rn} \cite{Sto11f} solutions. As in those simpler cases, it becomes clear that the singularity can coexist with the geometric and topological structures of the spacetime, in a way which doesn't destroy the information contained in the fields. As in the other cases, we can extrapolate for the case when the black hole is not eternal, \eg when it evaporates. This is because the {\kn} solution is, according to the no-hair theorem, representative for all kinds of black holes. 

The fact that the metric is allowed to become degenerate is not a problem, because, as shown in \cite{Sto11a,Sto11b,Sto11d}, we have now the mathematical apparatus to deal with this kind of singularities.

In conclusion, despites the singularities present inside the black holes, there is no reason to consider the {\kn} black holes destroy causality, the evolution equations and the information conservation. The {\kn} black holes are the most general stationary solution. The no-hair theorem makes them typical for our universe. They are typical even for the evaporating black holes, because the foliations presented here allow smooth modifications of the parameters $m$, $q$, and $a$, while preserving the topology. Moreover, we obtained charged singularities with smooth electromagnetic potential, leading to models of non-singular charged particles. This is why we can be more optimistic about the singularities of the general black holes as well.

\bibliographystyle{plain}%{unsrt}%{amsalpha}%{amsplain}

\end{document}